\documentclass[a4paper,conference]{IEEEtran}

\usepackage[
  a4paper,
  left=1.73cm,
  right=1.73cm,
  top=2.07cm,
  bottom=4.0cm
]{geometry}

\IEEEoverridecommandlockouts
\usepackage{cite}
\usepackage{amsmath,amssymb,amsfonts,amsthm}
\usepackage{graphicx}
\usepackage{booktabs}
\usepackage{tikz}
\usetikzlibrary{arrows.meta,positioning,fit}
\usepackage{microtype}
\usepackage{url}
\usepackage{balance}
\usepackage{comment}
\newtheorem{theo}{Theorem}
\newtheorem{prop}{Proposition}

\newtheorem{defi}{Definition}
\theoremstyle{remark}

\newtheorem{ex}{Example}

\newcommand{\X}{\mathcal{X}}
\newcommand{\C}{\mathcal{C}}
\newcommand{\Ssupp}{\mathcal{S}}
\newcommand{\Fset}{\mathcal{F}}

\usepackage{algorithm}
\usepackage{algorithmic}

\begin{document}
\title{Support-Aware Telemetry Compression for 5G Positioning via Conditional Conflict Graphs}

\author{%
\IEEEauthorblockN{Mohammad Reza Deylam Salehi}
\IEEEauthorblockA{\textit{T{\'e}l{\'e}com SudParis, Institut Polytechnique de Paris}\thanks{This work received financial support from the French government under the EU Important Project of Common European Interest on Microelectronics and Communication Technologies (IPCEI ME/CT).}\\
Palaiseau, France\\ mdeylamsalehi@telecom-sudparis.eu}
\and
\IEEEauthorblockN{Hakima Chaouchi}
\IEEEauthorblockA{\textit{T{\'e}l{\'e}com SudParis, Institut Polytechnique de Paris}\\ 
Institut Mines-Télécom\\
Palaiseau, France\\
hakima.chaouchi@imt.fr}
}
\maketitle

\begin{abstract}
\(5\)G positioning TRPs report quantized measurements to a Location Management Function (LMF), even when applications require only a coarse location region. We formulate this task as distributed zero-error function computation, where each TRP transmits an index that preserves the native service decision on a declared finite support. Sparse, nonrectangular positioning supports make independently constructed local characteristic-graph colorings not necessarily jointly decodable. We introduce a conditional conflict graph that characterizes valid single-terminal updates and develop an alternating DSATUR construction with an explicit global decoder check. In a controlled 3-TRP study with \(120\) range-equivalent bins, all codebooks achieve zero error on \(12{,}772\) native tuples. For service cells up to \(100\)~m, ideal rates are \(5.38\)--\(5.44\) bits/epoch/TRP, compared with \(6.91\) for raw reporting and \(6.56\)--\(6.87\) for a contiguous-interval baseline. On an independent six-base-station TDoA test set, native-support recurrence is \(83.77\%\). For \(2\)--\(8\)~m grids, feasible color-tuple screening accepts \(85.19\)--\(89.39\%\) of bursts, with \(0.81\)--\(1.05\%\) decision mismatch among accepted bursts. Two-stage fallback provides \(8.0\)--\(25.0\%\) ideal measurement-rate savings. The same support model provides a single-epoch anomaly screen that rejects infeasible report tuples, while held-out timing-bias experiments evaluate its sensitivity. At \(\Delta=100\)~m, expanding the controlled support by one native quantization bin in each TRP dimension covers all \(30{,}000\) held-out tuples while retaining \(16.4\%\) ideal-rate savings and zero decoder conflicts. This characterizes the tradeoff between compression, coverage, and anomaly sensitivity. Finally, a separate OAI prototype carries three TRP/color pairs through an experimental NRPPa extension, reducing the complete aligned-PER MeasurementResponse from \(69\) bytes for native UL-RTOA reporting to \(38\) bytes.
\end{abstract}


\section{Introduction}
\label{sec:intro}

$5$G New Radio (NR) supports UL-TDOA, UL-AoA, multi-RTT, and related positioning methods coordinated by the LMF~\cite{3gpp38305,italiano2025tutorial}. In an uplink timing procedure, a UE transmits a positioning SRS, and geographically separated TRPs obtain local quantities such as UL-RTOA. Their serving gNBs forward the measurements through NRPPa, with the AMF transparently transporting the NRPPa payload between the NG-RAN and the LMF~\cite{3gpp38455}. The physically relevant positioning anchors are the TRPs: sectors mounted on one tower provide little TDOA baseline, whereas geographically distributed TRPs may be controlled by a single logical gNB.

Existing signaling is measurement-oriented, whereas many applications are decision-oriented. A factory controller may need an aisle, an emergency service a search tile, and a mobility function a zone. Native measurements are not needed to be distinguished when they induce the same required service decision. This motivates communication for computing~\cite{witsenhausen1976zero, korner1971coding, orlitsky2001coding, deylam2025entropy, salehi2025non, salehi2023achievable}: each reporting point communicates enough to evaluate the target function rather than reconstructing the full source.

A direct use of ordinary local characteristic graphs is unsafe because feasible positioning tuples generally occupy a sparse subset of the Cartesian alphabet. Independently valid local colorings can then merge complete tuples that require different LMF outputs. We address this joint-decodability problem through our contributions:
\begin{itemize}
  \item We give a global zero-error criterion and a conditional conflict graph that exactly characterizes valid updates of one TRP codebook on an arbitrary finite support.
  \item We develop a support-aware alternating DSATUR construction for distributed reporters. Starting from native identity codebooks, the method updates one terminal at a time through graph coloring and verifies the global zero-error condition after every accepted update. We compare the resulting codebooks with a globally valid interval-constrained baseline and evaluate their sensitivity to the terminal update order.
  \item We evaluate the method on a controlled 3-TRP geometry and on independent six-base-station measured TDoA traces. We separate theorem-certified native-support recurrence from color-tuple acceptance, fallback, and out-of-support decision mismatch. We also quantify how a one-bin certified support enlargement trades compression for support coverage and static anomaly sensitivity.
  \item We implement the reporting path in OAI with a local experimental NRPPa extension. The prototype measures complete aligned-Packed Encoding Rules (PER) messages and validates gNB-to-LMF decoding with three logical TRPs.
\end{itemize}
The controlled study isolates the coding layer with deterministic range-equivalent timing bins. The measured-data study in Sec.~\ref{sec:trace-transfer} tests transfer to an independent downlink TDoA trajectory. Neither study is synchronized multi-TRP UL-RTOA RF validation. Section~\ref{sec:integration} therefore reports a separate OAI software/wire-path experiment. It uses three logical TRPs and a temporary synthetic codebook/decoder, and multi-radio RF/OTA validation remains outside the present scope.

{\bf Related Work and Scope}: Characteristic graphs originate in zero-error source coding with side information~\cite{witsenhausen1976zero, korner1971coding} and distributed function computation~\cite{orlitsky2001coding, alon1996source, salehi2023achievable, deylam2025entropy, salehi2025non, malak2024multi}. In this setting, the desired function determines which source values must remain distinguishable. This objective is also related to semantic and goal-oriented communication, where information generation and transmission are adapted to the significance of the data or to an application-level objective~\cite{uysal2022semantic, chiariotti2026goal}. Unlike approaches centered on timeliness, semantic distortion, or access decisions, we require exact recovery of a deterministic positioning service function over a declared finite support from independently encoded TRP reports. The positioning setting introduces geometry-dependent feasible tuples, distributed TRP encoders, and centralized LMF decoding. Related positioning studies include cooperative TDoA localization~\cite{marinho2021cooperative} and practical indoor ranging and localization~\cite{jin2022offshelf}, and practical 5G positioning using uplink SRS measurements, including multi-gNB SRS-based positioning under mixed propagation conditions ~\cite{dong2026deployable} and real-time indoor localization implemented over OpenAirInterface (OAI)~\cite{hsieh2026uplink}. Prior OAI studies provide the UL-TDOA/NRPPa path~\cite{malik2024concept}, SRS telemetry extraction~\cite{mundlamuri2024tools}, and testbed observations~\cite{ahadi2025experimental, ahadi2026tcml, bouknana2026ran}. We add a support-aware reporting and consistency layer without modifying the SRS procedure or the native positioning estimator.

{\bf Organization:} Sec.~\ref{sec:sys} introduce the zero-error system model. Sec.~\ref{sec:cond-graph} presents the codebook construction. Sec.~\ref{sec:evaluation} reports the controlled and trace-driven evaluations. Secs.~\ref{sec:integration} and~\ref{sec:anomaly} describe the OAI protocol prototype and the static anomaly screen, including certified support enlargement. Sec.~\ref{sec:conclu} discusses scope and concludes the paper.

\section{System Model and Zero-Error Requirement}
\label{sec:sys}

Consider one positioning epoch with $N$ reporting TRPs. TRP $i$ observes a native quantized symbol $X_i\in\X_i$, and the joint input $X=(X_1, X_2, \ldots,X_N)$ has finite feasible support
\begin{align}
  \Ssupp\subseteq \X_1\times \X_2\times\cdots\times\X_N\ .
\end{align}
The support incorporates the deployment geometry, quantization, and any nominal uncertainty included at design time. The uncompressed LMF applies a deterministic service function
\begin{align}
\label{eq:func-mapp}
  \ell:\Ssupp\rightarrow\mathcal{Y}\ ,
\end{align}
where $\mathcal{Y}$ may represent a grid cell, service zone, or alarm-relevant decision (see~Fig.~\ref{fig:sys}). The function $\ell$ includes the native positioning estimator and the subsequent projection to the service grid.

\begin{figure}[h]
\centering
\includegraphics[width=\columnwidth]{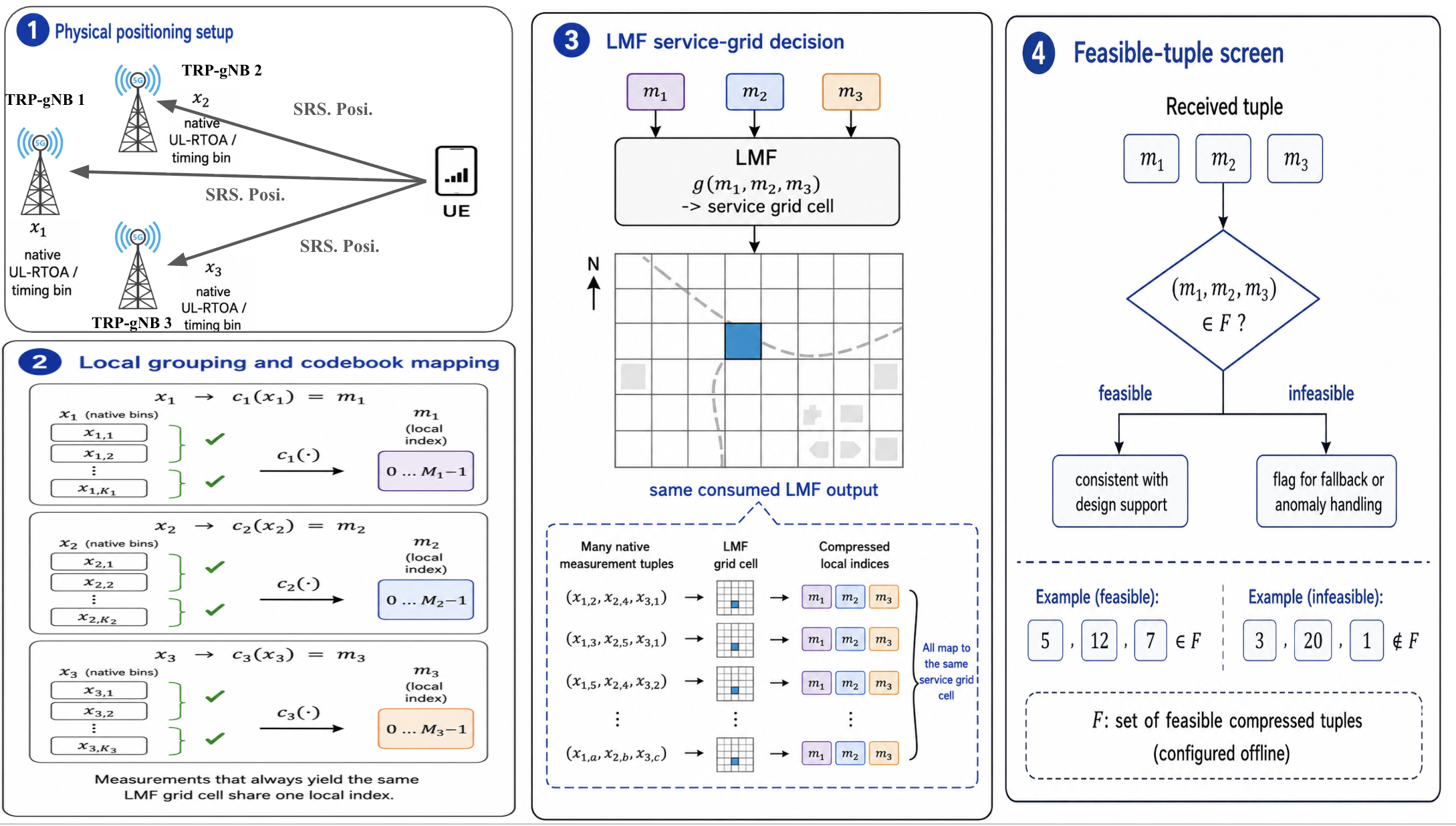}
\caption{System model of function-aware positioning telemetry. Native TRP measurements are quantized and mapped to local code indices. The LMF decodes the received index tuple into the same service-grid cell as the uncompressed pipeline. The feasible report-tuple set also enables single-epoch anomaly screening.}
\label{fig:sys}
\end{figure}

TRP $i$ sends $M_i=c_i(X_i)$ from a finite code alphabet $\C_i$. Define the aggregate report map
\begin{align}
\label{eq:aggregate-map}
 \mathbf{c}(x)\triangleq\bigl(c_1(x_1), c_2(x_2),\ldots,c_N(x_N)\bigr).
\end{align}
A decoder $g:\C_1\times\cdots\times\C_N\rightarrow\mathcal{Y}$ is zero-error on $\Ssupp$ when
\begin{align}
\label{eq:zero-error}
 g\bigl(\mathbf{c}(x)\bigr)=\ell(x),\quad \forall x\in\Ssupp.
\end{align}
The ideal alphabet rate is $\log_2|\C_i|$ bits per epoch, while a fixed-width binary representation requires $\lceil\log_2|\C_i|\rceil$ bits. If the distribution of $c_i(X_i)$ is known, a binary prefix code can achieve an expected length below $H(c_i(X_i))+1$ bits~\cite{cover_thomas}. Any rate or entropy optimization must still preserve the global validity condition defined below.
\begin{prop}[Global product-partition criterion]
\label{prop:global}
The local encoders $(c_1, c_2,\ldots,c_N)$ admit a decoder satisfying~\eqref{eq:zero-error} if and only if, for every pair of nodes $x,x'\in\Ssupp$,
\begin{align}
\label{eq:global}
 \mathbf{c}(x)=\mathbf{c}(x')
 \quad\rightarrow\quad
 \ell(x)=\ell(x').
\end{align}
\end{prop}
\begin{proof}
First, suppose that a zero-error decoder $g$ exists. If $\mathbf{c}(x)=\mathbf{c}(x')$, then
\begin{align*}
 \ell(x)=g\bigl(\mathbf{c}(x)\bigr)
 =g\bigl(\mathbf{c}(x')\bigr)=\ell(x'),
\end{align*}
which proves~\eqref{eq:global}.

Conversely, suppose that~\eqref{eq:global} holds. For every reachable report tuple $m\in\mathbf{c}(\Ssupp)\triangleq\{\mathbf{c}(x):x\in\Ssupp\}$, define $g(m)=\ell(x)$ for any $x\in\Ssupp$ satisfying $\mathbf{c}(x)=m$. Condition~\eqref{eq:global} guarantees that this definition is well defined. The value of $g$ on report tuples outside $\mathbf{c}(\Ssupp)$ may be assigned arbitrarily. The resulting decoder~satisfies~\eqref{eq:zero-error}, thus proves the proposition.
\end{proof}

The support structure matters. Consider $\Ssupp=\{(0,0),(1,1)\}$ with $\ell(0,0)\neq\ell(1,1)$. For each terminal, the ordinary characteristic graph contains no edge between its symbols $0$ and $1$, because the two symbols never occur with a common value of the other terminal. Thus, a one-color assignment is proper for each local graph. However, these assignments yield $\mathbf{c}(0,0)=\mathbf{c}(1,1)$ although the corresponding function values differ, violating~\eqref{eq:global}. Independent proper coloring of ordinary local characteristic graphs is therefore insufficient on general non-rectangular supports.

\section{Conditional Conflict Graph Coding}
\label{sec:cond-graph}

Fix all codebooks except that of terminal $i$. Let $x_{-i}\triangleq(x_j)_{j\neq i}$ and $c_{-i} \triangleq(c_j)_{j\neq i}$ denote, respectively, the other local symbols and encoders. Define their aggregate report map by
\begin{align}
\label{eq:aggregate-minus-i}
 \mathbf{c}_{-i}(x_{-i})\triangleq\bigl(c_j(x_j)\bigr)_{j\neq i}\ .
\end{align}

\begin{defi}[Conditional conflict graph]
\label{def:conditional}
For fixed $c_{-i}$, define the graph $H_i(c_{-i})=G(\X_i,E_i)$, which may contain self-loops. An edge joins $a,b\in\X_i$ if there exist $x,x'\in\Ssupp$ such that $x_i=a$, $x_i'=b$, and
\begin{align}
\label{eq:theo-1-pair}
 \mathbf{c}_{-i}(x_{-i})=\mathbf{c}_{-i}(x_{-i}'),
 \quad \ell(x)\neq\ell(x').
\end{align}
If $H_i(c_{-i})$ contains a self-loop, no choice of $c_i$ can make the complete encoder collection zero-error while $c_{-i}$ remains fixed. Otherwise, $c_i$ is proper if $c_i(a)\neq c_i(b)$ for every edge joining distinct vertices $a$ and $b$.
\end{defi}
With identity codebooks at all other terminals, $H_i(c_{-i})$ reduces to the ordinary characteristic graph. The next theorem gives the exact validity condition for one-terminal updates.

\begin{theo}[Exact one-terminal update]
\label{thm:update}
For fixed $c_{-i}$, the complete encoder collection $(c_i,c_{-i})$ is zero-error on $\Ssupp$ if and only if $H_i(c_{-i})$ has no self-loop and $c_i$ is a proper coloring of its non-loop edges.
\end{theo}
\begin{proof}
If $c_i$ assigns the same color to the endpoints of a conflict edge, the corresponding support points $x,x'$ in Definition~\ref{def:conditional} satisfy $\mathbf{c}(x)=\mathbf{c}(x')$ and $\ell(x)\neq\ell(x')$, violating Proposition~\ref{prop:global}. A self-loop yields the same contradiction for every possible choice of $c_i$.

Conversely, suppose that $(c_i,c_{-i})$ is not zero-error. By Proposition~\ref{prop:global}, there exist $x,x'\in\Ssupp$ such that $\mathbf{c}(x)=\mathbf{c}(x')$ and $\ell(x)\neq\ell(x')$. Let $a=x_i$ and $b=x_i'$. Equality of the reports from all other terminals implies that $a$ and $b$ are adjacent in $H_i(c_{-i})$, while equality of the complete report tuples implies that $c_i(a)=c_i(b)$. If $a=b$, then $H_i(c_{-i})$ contains a self-loop. If $a\neq b$, two adjacent vertices receive the same color, contradicting properness of $c_i$.
\end{proof}

\subsection{Alternating construction and verification}
\label{sec:alt-const-verfi}

The following example shows why the conflict graph must depend on the current codebooks of the other TRPs.

\begin{ex}[Conditional conflict graph]
Consider $N=3$ TRPs with native alphabets
$\mathcal X_1=\mathcal X_2=\mathcal X_3=\{0,1,\dots,7\}$. The declared support is
\begin{align}
\label{eq-ex-1}
\mathcal S=\{&(0,0,0),(1,1,0),(0,2,1),(1,3,1), (2,4,0)\nonumber\\&,(3,5,0)\},
\end{align}
with labels
\begin{align*}
\ell(0,0,0)&=A, & \ell(1,1,0)&=B,\\
\ell(0,2,1)&=A, & \ell(1,3,1)&=B,\\
\ell(2,4,0)&=C, & \ell(3,5,0)&=C.
\end{align*}
Suppose TRP~2 already uses
$$
c_2(0)=c_2(1)=\alpha,\:\:
c_2(2)=c_2(3)=\beta,\:\:
c_2(4)=c_2(5)=\gamma,
$$
while TRP~3 uses the identity map on its active symbols: $c_3(0)=0$ and $c_3(1)=1$. We now update TRP~1.

With identity codebooks at the other terminals, the ordinary local graph at TRP~1 has no edge between symbols $0$ and $1$: no two support points share the same $(x_2,x_3)$ and have different labels. After compressing TRP~2, however,
$$(0,0,0)\ \text{and}\ (1,1,0),\quad
(0,2,1)\ \text{and}\ (1,3,1)$$
produce identical reports at the other terminals:
$$(c_2(0),c_3(0))=(c_2(1),c_3(0))=(\alpha,0),$$
and
$$(c_2(2),c_3(1))=(c_2(3),c_3(1))=(\beta,1).$$
Their service labels differ, so $\{0,1\}$ is an edge of $H_1(c_{-1})$. The pair $(2,4,0)$ and $(3,5,0)$ has the common label $C$ and induces no edge. Hence the active graph on $\{0,1,2,3\}$ has one nontrivial edge, $\{0,1\}$. One valid update is
$$c_1(0)=r,\quad c_1(1)=b,\quad c_1(2)=c_1(3)=g.$$
Thus, a merge that is safe before another TRP is compressed can become invalid after that TRP's report is coarsened.
\end{ex}

\begin{figure}[h]
\centering
\begin{tikzpicture}[
    vertex/.style={
        circle,
        draw,
        minimum size=6.5mm,
        inner sep=0pt,
        font=\small
    }
]

\node at (0,1.4)
{\textbf{Conditional conflict graph $H_1(c_{-1})$}};

\node[vertex, fill=red!30]   (v0) at (-1, 0.5) {$0$};
\node[vertex, fill=blue!30]  (v1) at ( 1, 0.5) {$1$};
\node[vertex, fill=green!30] (v2) at (-1,-0.8) {$2$};
\node[vertex, fill=green!30] (v3) at ( 1,-0.8) {$3$};

\draw[thick] (v0) -- (v1);

\end{tikzpicture}
\caption{Conditional conflict graph for TRP~1 in Example~1. The edge $\{0,1\}$ forces distinct colors; the nonadjacent symbols 2 and 3 may share one.}
\label{fig:toy-conflict}
\end{figure}

Identity codebooks satisfy Proposition~\ref{prop:global} and provide the initialization. For a selected terminal order, Algorithm~\ref{alg:construction} updates one codebook at a time and accepts an update only after a global decoder check.

\begin{algorithm}[h]
\caption{Support-aware alternating codebook construction}
\label{alg:construction}
\scriptsize
\begin{algorithmic}[1]
\REQUIRE Support $\Ssupp$, labels $\ell$, terminal order $\pi$,
number of passes $P$
\ENSURE Codebooks $(c_1,\ldots,c_N)$, decoder $g$, feasible set
$\Fset$
\STATE Initialize $c_i(x_i)\gets x_i$ for all $i$ and $x_i\in\X_i$
\FOR{$p=1,\ldots,P$}
  \FOR{$i$ in order $\pi$}
    \STATE Group $x\in\Ssupp$ by $\mathbf c_{-i}(x_{-i})$
    \STATE Construct $H_i(c_{-i})$ from output disagreements
    \IF{$H_i(c_{-i})$ has no self-loop}
      \STATE Color $H_i(c_{-i})$ by DSATUR to obtain
      $\widetilde c_i$
      \IF{$|\operatorname{im}\widetilde c_i|\leq|\operatorname{im}c_i|$
      and $(\widetilde c_i,c_{-i})$ satisfies Proposition~\ref{prop:global}}
        \STATE $c_i\gets\widetilde c_i$
      \ENDIF
    \ENDIF
  \ENDFOR
\ENDFOR
\STATE Construct $g$ from $\mathbf c(\Ssupp)$
\STATE $\Fset\gets\mathbf c(\Ssupp)$
\RETURN $(c_1,\ldots,c_N),g,\Fset$
\end{algorithmic}
\end{algorithm}

We run Algorithm~\ref{alg:construction} for all six permutations of the three terminals and retain the valid construction with the smallest sum of ideal rates. DSATUR~\cite{brelaz1979new} is heuristic, so the reported $\mathcal{K}_i$ values are achievable codebook sizes, not chromatic numbers or global optima. Every order uses the same $P=3$ pass budget. The order comparison is therefore a fixed-budget robustness test, not a convergence claim.

After construction, TRP~$i$ stores only its local map $c_i$. The LMF stores the joint decoder $g$ and the feasible report set $\Fset=\mathbf c(\Ssupp)$. TRPs encode independently and exchange no measurements. The contiguous-interval baseline uses the same conditional graphs but restricts each color class to a consecutive interval of native bins. It greedily selects the longest conflict-free prefix and applies the same global decoder check after every update.

\section{Reproducible Positioning Study}
\label{sec:evaluation}

\subsection{Geometry, quantization, and LMF function}
\label{subsec:geo-quant}

Three spatially separated TRPs are placed at $(0,0)$, $(1000,0)$, and $(0,1000)$ m. A $220\times220$ grid of UE locations forms the design set. Each propagation distance is quantized into $M=120$ bins over $[0,\sqrt{2}\,\mathrm{km}]$, giving an $11.8$ m range-equivalent timing resolution and raw ideal rate $\log_2(120)=6.907$ bits/TRP.

For native bin centers $(r_1,r_2,r_3)$, the LMF computes
\begin{align}
\label{eq-position-grid-extension}
 \widehat p_x=\frac{r_1^2-r_2^2+10^6}{2000},\quad
 \widehat p_y=\frac{r_1^2-r_3^2+10^6}{2000}\ ,
\end{align}
projects the estimate onto the $1\,\mathrm{km}\times1\,\mathrm{km}$ service region, and assigns it to a square cell of length $\Delta$. The $48,400$ sampled locations produce $12,772$ distinct native measurement tuples. This synchronized range-equivalent model isolates telemetry coding on a finite alphabet. It omits the common UE transmission-time term, receiver-clock bias, NLOS effects, and SRS estimation errors of a complete UL-TDOA model.

\subsection{Controlled compression results}
Fig.~\ref{fig:rate} and Table~\ref{tab:rate} report average ideal rates. Every proposed and baseline point is validated by constructing the complete report-tuple decoder over all $12,772$ tuples, and the conflict count is zero throughout.

\begin{figure}[t]
  \centering
  \includegraphics[width=\columnwidth]{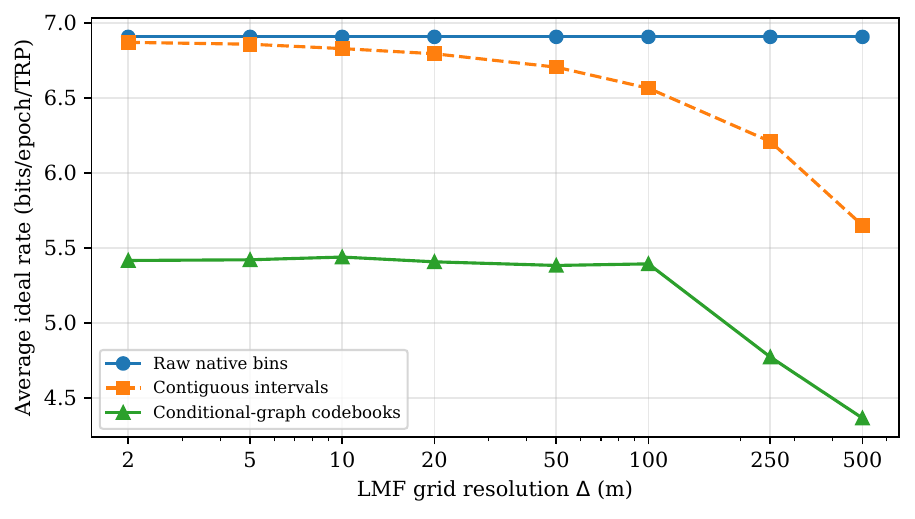}
  \caption{Globally validated ideal reporting rates versus service-grid resolution.}
  \label{fig:rate}
\end{figure}

\begin{table}[t]
\centering
\caption{Validated function-aware codebooks and rates.}
\label{tab:rate}
\scriptsize
\setlength{\tabcolsep}{2.7pt}
\begin{tabular}{@{}rcrrr@{}}
\toprule
$\Delta$ (m) & $(\mathcal{K}_1,\mathcal{K}_2,\mathcal{K}_3)$ & $R_{\rm fun}$ & $R_{\rm int}$ & save raw \\
\midrule
2   & $(109,7,102)$ & 5.416 & 6.870 & 21.6\%\\
5   & $(109,7,103)$ & 5.421 & 6.858 & 21.5\%\\
10  & $(110,7,106)$ & 5.439 & 6.829 & 21.3\%\\
20  & $(107,7,102)$ & 5.407 & 6.794 & 21.7\%\\
50  & $(107,7,97)$  & 5.383 & 6.705 & 22.1\%\\
100 & $(103,7,103)$ & 5.393 & 6.565 & 21.9\%\\
250 & $(82,3,83)$   & 4.773 & 6.209 & 30.9\%\\
500 & $(54,3,54)$   & 4.365 & 5.651 & 36.8\%\\
\bottomrule
\end{tabular}
\vspace{-0.5cm}
\end{table}

For $\Delta\leq100$~m, function-aware coding saves $21.3$--$22.1\%$ of the ideal raw rate. The average fixed-field width is $5.67$ bits/TRP, a $19.0\%$ reduction from the $7$-bit native field. The interval baseline stays close to raw reporting because many valid color classes are noncontiguous.

The gain is not only a cardinality effect. Under the empirical distribution of all $48{,}400$ design locations, the average marginal entropy of the native symbols is $6.592$ bits/TRP. At $\Delta=100$~m, the corresponding average entropy of the three color variables is $5.108$ bits/TRP, a $22.5\%$ reduction. This comparison still concerns separate per-TRP source descriptions; it is not a Slepian-Wolf sum-rate claim.

The selected order can make the per-TRP allocation asymmetric. At $\Delta=100$~m, order $(2,1,3)$ gives $(\mathcal{K}_1,\mathcal{K}_2,\mathcal{K}_3)=(103,7,103)$. TRP~2 is updated while the other codebooks are still identities, so its conditional graph has the finest side context and permits more merging. Later updates see coarser contexts and denser conflict graphs. TRPs~2 and~3 are geometrically symmetric here. Therefore, the allocation difference is caused mainly by update order and heuristic coloring.

We evaluate all $3!=6$ orders with the same $P=3$ pass budget. At $\Delta=100$~m, four orders achieve $R_{\rm fun}=5.393$ bits/epoch/TRP, whereas $(1,2,3)$ and $(1,3,2)$ give $5.582$ bits/epoch/TRP. The best to worst difference is $3.49\%$. Representative allocations are $(103,7,103)$ and $(33,45,74)$, so the local allocation varies more than the aggregate rate. Over $\Delta\leq100$~m, the rate spread across the six orders is $2.19$--$7.61\%$. Every order passes the global decoder test. We therefore report the minimum-rate valid construction over all six orders. This does not imply convergence or global optimality.

The controlled geometry isolates the coding mechanism. We next test whether the learned finite-support structure recurs on independent measured TDoA trajectories.

\subsection{Trace-driven generalization and fallback}
\label{sec:trace-transfer}

We apply the construction to the Fraunhofer IIS indoor $5$G dataset, which provides synchronized TDoA measurements from six downlink base stations and separate training and test trajectories~\cite{fraunhofer5gtdoa}. Each measurement is quantized to $B=24$ symbols. Quantizers, support, codebooks, and decoder are built only from the $18{,}863$ training bursts. A deterministic fingerprinting estimator maps each native tuple to a two-dimensional service-grid cell. Before service-grid projection, its native quantized test-position error is $0.49$~m median and $1.17$~m at the 90th percentile.

Among the $15{,}722$ independent test bursts, $83.77\%$ reproduce a native tuple contained in the training support. These bursts form the \emph{certified subset}: by Proposition~\ref{prop:global}, their compressed and native reports necessarily produce the same service decision. The remaining $16.23\%$ contain native tuples that were not observed in the training support, so the zero-error theorem makes no claim for them. We therefore evaluate the deployed decoder on the complete test set rather than discarding those bursts.

This residual mass is treated as an operating-region issue rather than silently assigned a compressed decision. The safe policy is native fallback whenever the compressed tuple is infeasible. Therefore, a deployment can additionally provision a broader, independently certified support/codebook version and activate it coherently when sustained fallback indicates support drift. Section~\ref{sec:anomaly} quantifies this coverage-rate trade-off in the controlled model without using test data to redefine the measured-trace support.

For each received color tuple $m=(m_1,\ldots,m_6)$, the LMF first tests membership in the training-derived feasible set $\Fset=\mathbf{c}(\Ssupp)$. For $\Delta=2,4,8$~m, respectively, $85.19\%$, $86.31\%$, and $89.39\%$ of all test bursts produce a tuple in $\Fset$. The remaining $14.81\%$, $13.69\%$, and $10.61\%$ are rejected by the feasibility test and trigger native fallback. Membership in $\Fset$ is weaker than native-support membership: an unseen native tuple can still map to a color tuple already contained in $\Fset$. Accordingly, among accepted test bursts, the empirical compressed/native service-decision mismatch is $1.008\%$, $0.811\%$, and $1.053\%$ for $\Delta=2,4,8$~m, respectively. These mismatches occur outside the theorem-certified regime and therefore do not contradict the zero-error guarantee on $\Ssupp$.

Table~\ref{tab:trace-transfer} reports both the nominal function-code rate and a conservative two-stage fallback accounting. If an infeasible compressed tuple is followed by retransmission of the six native quantized measurements, the ideal measurement rate per base station is
\begin{align}
\label{eq-rate-2stage}
R_{\rm 2stage}=R_{\rm fun}+P_{\rm fb}R_{\rm raw},
\end{align}
where $P_{\rm fb}=\Pr\{\mathbf{c}(X)\notin\Fset\}$. This accounting assumes that native values are retained until fallback is resolved. It excludes request/response and ASN.1/PER overhead. This section therefore evaluates coding on measured downlink TDoA only, and NRPPa serialization is evaluated separately in Sec.~\ref{sec:integration}.

\begin{table}[t]
\centering
\caption{Measured six-BS evaluation on all $15{,}722$ independent test bursts. Native-support recurrence is $83.77\%$ for every service-grid resolution.}
\label{tab:trace-transfer}
\scriptsize
\setlength{\tabcolsep}{2.4pt}
\begin{tabular}{@{}rrrrrrr@{}}
\toprule
$\Delta$ &
$R_{\rm fun}$ &
$P_{\rm acc}$ &
$P_{\rm fb}$ &
$P_{\rm mis|acc}$ &
$R_{\rm 2stage}$ &
Saving \\
(m) &
(bits/BS) &
\multicolumn{3}{c}{(\%)} &
(bits/BS) &
(\%) \\
\midrule
2 & 3.538 & 85.19 & 14.81 & 1.008 & 4.217 & 8.0 \\
4 & 3.215 & 86.31 & 13.69 & 0.811 & 3.843 & 16.2 \\
8 & 2.954 & 89.39 & 10.61 & 1.053 & 3.440 & 25.0 \\
\bottomrule
\end{tabular}
\vspace{-0.35cm}
\end{table}

Here $R_{\rm raw}=\log_2 24=4.585$~bits/epoch/BS, $P_{\rm acc}=\Pr\{\mathbf{c}(X)\in\Fset\}$, $P_{\rm fb}=1-P_{\rm acc}$, and $P_{\rm mis|acc}$ is the empirical service-decision mismatch conditioned on color-tuple acceptance. The $R_{\rm 2stage}$ column includes native retransmission after a feasibility rejection but not control-plane or serialization overhead.

\section{NRPPa Integration and OAI Protocol Prototype}
\label{sec:integration}

Fig.~\ref{fig:flow} maps the reporting layer onto the UE--TRP/gNB--AMF--LMF transaction. Native mode carries UL-RTOA results. Compressed mode maps each local symbol to a color index and sends a codebook version with explicit $(\mathrm{TRP\mbox{-}ID},\mathrm{colorIndex})$ pairs. The LMF then evaluates the service output from the received color tuple. The prototype validates this software and wire path with logical TRPs, and it is not a synchronized multi-radio experiment.
\begin{figure}[t]
  \centering
  \resizebox{\columnwidth}{!}{
\begin{tikzpicture}[
  actor/.style={font=\bfseries\small},
  lifeline/.style={black!40,thin},
  existing/.style={-{Latex[length=2mm]},black,thick},
  ours/.style={-{Latex[length=2mm]},blue!65!black,very thick},
  lbl/.style={font=\scriptsize,align=center},
  ourslbl/.style={font=\scriptsize\bfseries,align=center,blue!65!black}
]
\def\ytop{8.6}
\def\ybot{0.0}
\foreach \x/\name in {0/UE, 3.1/TRP--gNB, 6.2/AMF, 9.0/LMF} {
  \node[actor] at (\x,\ytop+0.4) {\name};
  \draw[lifeline] (\x,\ytop) -- (\x,\ybot);
}
\draw[existing] (9.0,8.0) -- (6.2,8.0) node[midway,above,lbl]{Positioning Information Request};
\draw[existing] (6.2,7.55) -- (3.1,7.55) node[midway,above,lbl]{UE-associated NGAP/NRPPa transport};
\draw[existing] (3.1,6.9) -- (0,6.9) node[midway,above,lbl]{SRS resource configuration (RRC)};
\draw[existing] (0,6.3) -- (3.1,6.3) node[midway,above,lbl]{Positioning SRS transmission};
\draw[existing] (3.1,5.9) to[out=20,in=-20,looseness=6] (3.1,5.5);
\node[lbl,right] at (3.35,5.7) {PHY: SRS estimate $\rightarrow$ native UL-RTOA $x_i$};
\node[lbl,right,black!60] at (3.35,5.30) {(e.g., FAPI indication / \texttt{T\_tracer})};
\draw[ours] (3.1,4.9) to[out=20,in=-20,looseness=6] (3.1,4.5);
\node[ourslbl,right] at (3.35,4.7) {codebook lookup $c_i(x_i)\rightarrow m_i$};
\node[lbl,right,black!60] at (3.35,4.30) {(codebook ID, color index, status)};
\draw[ours] (3.1,3.7) -- (6.2,3.7) node[midway,above,ourslbl]{Measurement Report: experimental extension $[\mathrm{ID},m_i]$};
\draw[ours] (6.2,3.25) -- (9.0,3.25) node[midway,above,lbl,blue!65!black]{transparent AMF transport};
\draw[ours] (9.0,2.65) to[out=20,in=-20,looseness=6] (9.0,2.25);
\node[ourslbl,left] at (8.75,2.45) {decode $g(m_1,\ldots,m_N)\rightarrow$ service cell};
\node[lbl,left,black!60] at (8.75,1.95) {test feasible tuple $m\in\mathcal F$; otherwise flag/fallback};
\node[lbl,black,anchor=west] at (0,0.9) {\textbf{---} existing 3GPP positioning procedure};
\node[ourslbl,anchor=west] at (0,0.4) {\textbf{---} proposed function-aware reporting operations};
\end{tikzpicture}}
  \caption{Integration of the proposed mechanism into the existing NRPPa UL-TDOA procedure. Black arrows indicate existing positioning signaling, while blue arrows and labels indicate the proposed reporting, decoding, and anomaly-screening operations.}
  \label{fig:flow}
\end{figure}

\subsection{Protocol semantics and safeguards}
\label{sec:prot-3-safe-gurd}

A color index is not a UL-RTOA value. Therefore, we keep native and fallback reports in the standard NRPPa format and implement compressed mode as a local non-standard schema understood by both modified endpoints. A native MeasurementResponse carries the standard TRP measurement-response list. The compressed variant instead carries local IE~65000 with a codebook version and exactly three explicit $(\mathrm{TRP\mbox{-}ID},\mathrm{colorIndex})$ pairs. Explicit IDs remove any dependence on list order. Because the variant changes the native mandatory-field semantics, the $38$-byte message below is experimental and not a standards-compliant UL-RTOA replacement.

Fallback rules follow the coding model. An unknown codebook version or an out-of-table local symbol selects native reporting. A received color tuple outside $\Fset$ is rejected or followed by native fallback. The prototype validates one compressed transaction and one forced native baseline; it does not exhaustively test malformed or stale-version cases. Byte counts are taken after aligned-PER serialization, not inferred from $\log_2 \mathcal{K}_i$.

\subsection{Protocol implementation and wire-format measurement}
\label{sec:oai-protocol-poc}

The software experiment uses the OAI RAN and OAI-CN5G LMF code bases~\cite{oai_ran_repo, oai_lmf_repo}. The RAN is pinned to tag \texttt{2026.w40}; the LMF starts from revision \texttt{d95a651c}. 
The modified gNB test harness creates three logical TRPs, IDs $1$, $2$, and $3$, with one standardized NRPPa UL-RTOA K1 result per TRP, as defined in 3GPP TS~38.455 (ETSI TS~138~455, Sec.~9.2.39)~\cite{3gpp38455}. The tested native tuple is $(98502,98504,98506)$, whose entries lie within the standardized K1 range $0,\ldots,985025$. Temporary lookup tables map it to color tuple $(1,2,1)$ under codebook version~1. The lookup tables and the one-entry LMF decoder are synthetic. They validate endpoint logic, serialization, transport, and receive-side decoding only.

Table~\ref{tab:oai-wire} gives the complete aligned-PER sizes for the same K1 inputs. The native MeasurementResponse contains three UL-RTOA results with their timestamps and response context and uses $69$ bytes. The experimental compressed response carries the version and three explicit TRP/color pairs and uses $38$ bytes. The difference is $31$ bytes, or $44.9\%$. This wire result is separate from the ideal coding rates in Sec.~\ref{sec:evaluation} and Sec.~\ref{sec:trace-transfer}. The comparison is function-equivalent for the demonstrated service transaction rather than field-equivalent: the local compressed schema intentionally does not reproduce every native measurement-metadata field.

\begin{table}[t]
\centering
\caption{Protocol-level OAI byte accounting for the same three logical UL-RTOA K1 inputs.}
\label{tab:oai-wire}
\scriptsize
\setlength{\tabcolsep}{3.0pt}
\renewcommand{\arraystretch}{1.05}
\begin{tabular}{@{}p{0.20\columnwidth}p{0.52\columnwidth}r@{}}
\toprule
Mode & Complete MeasurementResponse content & APER bytes \\
\midrule
Native & Standard list with three UL-RTOA K1 results & 69 \\
Compressed & Local version + three explicit TRP/color pairs & 38 \\
\midrule
Reduction & Same transaction inputs & 31 (44.9\%) \\
\bottomrule
\end{tabular}
\end{table}

The experimental message traverses the OAI gNB--AMF--LMF path. The modified LMF decodes the $38$-byte APER PDU, recovers $(1,1),(2,2),(3,1)$, accepts codebook version~1, and reaches the synthetic decoder entry for tuple $(1,2,1)$. The HTTP/2 location request then returns a successful service response. Fig.~\ref{fig:oai-evidence} shows selected evidence from both runs. This is a protocol/software result: the experiment has no $3$ synchronized physical radios, and the synthetic service cell is not a positioning-accuracy result.

\begin{figure}[t]
  \centering
  \includegraphics[width=\columnwidth]{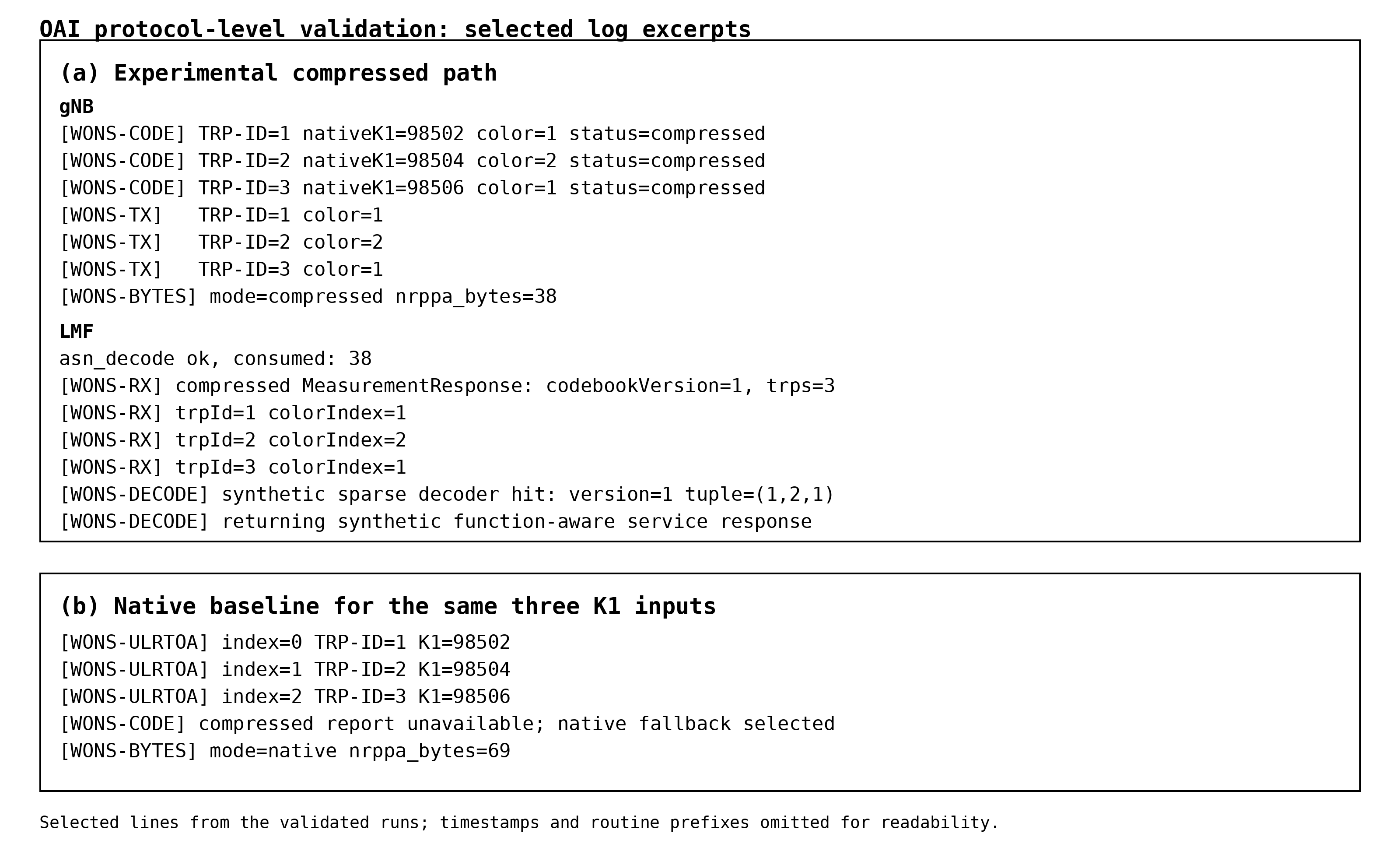}
  \caption{Selected logs from the validated compressed and native OAI runs. Full patches, scripts, hashes, and evidence are available at~\cite{oai_lmf_repo}. 
  The temporary codebook/decoder is protocol-only.}
  \label{fig:oai-evidence}
\end{figure}

\subsection{Timing model and RF scope}

The controlled study uses synchronized range-equivalent values to isolate the coding layer. A real uplink timing observation is
\begin{align}
\label{eq:toa-model}
 t_i=t_0+\frac{\|p-a_i\|}{c_0}+b_i+n_i\ ,
\end{align}
where $t_0$ is the unknown UE transmission-time term, $b_i$ is residual TRP clock/receiver bias, and $n_i$ collects estimation and propagation error. The encoder operates directly on the native quantized UL-RTOA field. The function $\ell$ retains the native LMF differencing, compensation, positioning, and grid-projection operations.

A deployed codebook must be bound to the measurement definition, active TRPs, support model, service function, and grid resolution. Unknown or mixed versions require native fallback. The present prototype measures serialized APER bytes. RF validation must additionally measure synchronized uplink epochs, latency, fallback frequency, and raw/compressed service agreement with the scientific codebooks driven by UL-RTOA or trace-replayed UL-RTOA.

\subsection{Complexity and stored state}
\label{sec:imple-foot-reproduc}

For one terminal update, support points are grouped by the current report tuple of the other terminals. If the group sizes are $n_1,n_2,\ldots$, graph construction examines $O(\sum_q n_q^2)$ candidate pairs rather than $O(|\Ssupp|^2)$ pairs. Each accepted update is followed by an $O(|\Ssupp|)$ global decoder check. Online operation uses one local array lookup at each TRP, one sparse decoder lookup at the LMF, and feasible-set membership testing.

The stored decoder is sparse. At $\Delta=100$~m, each local codebook contains $120$ native-bin entries and $(\mathcal{K}_1,\mathcal{K}_2,\mathcal{K}_3)=(103,7,103)$. Only $|\Fset|=12{,}180$ of the $103\times7\times103=74{,}263$ possible color tuples are reachable. These controlled-study statistics are distinct from the temporary one-entry decoder used in the OAI protocol experiment.

\section{Static Geometric Anomaly Screen}
\label{sec:anomaly}

Once the codebooks are fixed, define $\Fset\triangleq\mathbf{c}(\Ssupp)$. The zero-error decoder is constrained only on $\Fset$. Before decoding, the LMF can therefore reject any report tuple outside this set. Such a tuple cannot arise from the design support. It may indicate a codebook mismatch, receiver fault, large NLOS bias, stale report, or injected telemetry outside the feasible set. The screen cannot detect an anomalous tuple that remains in $\Fset$.

A positive timing bias models NLOS excess delay or receiver calibration drift. A report from a previous epoch models stale telemetry or an association error. A wrong codebook version models control-plane desynchronization. Missing, duplicated, and out-of-alphabet reports are explicit protocol faults and are checked before geometric membership. Detecting feasible but suspicious tuples requires temporal mobility checks or physical-layer authentication.

The same timing-domain symptom can also arise under adversarial delay manipulation. Prior work has studied physical-layer threats to 5G ToA positioning~\cite{focarelli2024toaThreats} and experimentally demonstrated full-frame 5G meaconing that biases ToA through delayed replay while communication remains operational~\cite{zanini2025meaconing}. Our injected range bias is only a measurement-domain analogue of such delay manipulation, not a waveform-level meaconing experiment. We therefore keep the broader term \emph{anomaly screen}: membership in $\Fset$ can flag an incompatible timing tuple, but it neither identifies the cause nor detects a malicious or faulty tuple that remains feasible.

We draw $30,000$ held-out UE positions and add a positive range-equivalent bias to TRP~1 before quantization. The same held-out locations are reused across bias levels, and biased measurements are saturated at the largest native bin. Fig.~\ref{fig:anomaly} reports the empirical flag probability. With no injected bias, the nominal support-miss rates are $2.06\%$, $1.70\%$, and $0.32\%$ for $\Delta=5$, $100$, and $500$ m. These values arise from finite support discretization and should not be interpreted as RF false-alarm rates. At a $24$ m bias, the corresponding flag probabilities are $53.9\%$, $39.2\%$, and $4.9\%$.

\begin{figure}[t]
  \centering
  \includegraphics[width=0.85\columnwidth]{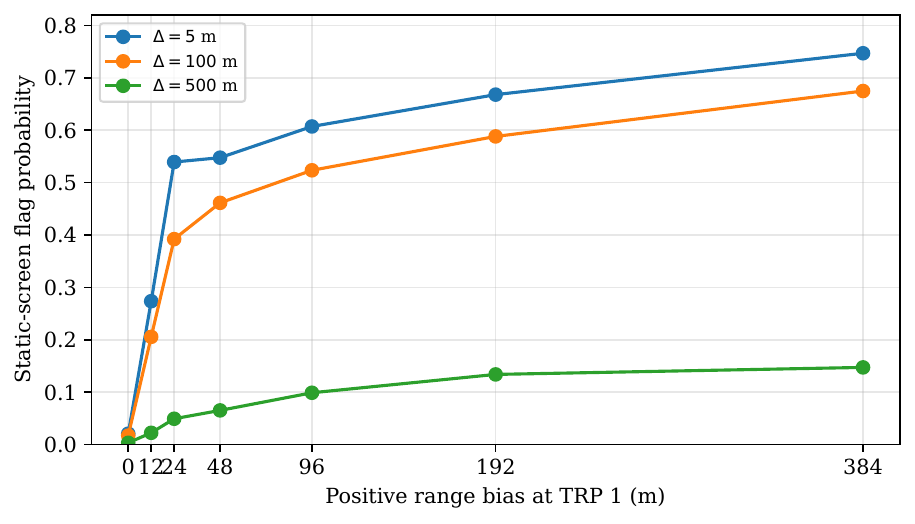}
  \caption{Held-out static-screen response to a positive range-equivalent bias at one TRP. The zero-bias point is the nominal support-miss rate.}
  \label{fig:anomaly}
\end{figure}

The result shows a rate-observability trade-off: coarser service functions lower both reporting rate and sensitivity to bias. The zero-bias point is a finite-support miss rate, not an RF false-alarm probability. A deployment should enlarge $\Ssupp$ with calibrated timing uncertainty or high-coverage traces, then estimate $P_{\rm FA}=P(\mathrm{flag}\mid\mathrm{nominal})$ and $P_{\rm D}=P(\mathrm{flag}\mid\mathrm{fault})$ for the target fault model. Temporal checks can then handle feasible but suspicious sequences.

\subsection{Certified support enlargement}
The previous experiment deliberately uses the nominal design support. To quantify the cost of making that support more tolerant, define the one-bin envelope
\begin{align}
\label{eq:support-envelope}
\Ssupp^{(1)}\!\triangleq\!\left\{z\in\X_1\!\times\!\X_2\!\times\!\X_3:\min_{x\in\Ssupp}\|z-x\|_\infty\leq1\right\}.
\end{align}
Clipping is applied at alphabet boundaries. For each added tuple, the same native estimator in~\eqref{eq-position-grid-extension} assigns the service label, after which Algorithm~\ref{alg:construction} is rerun from identity codebooks. Thus Proposition~\ref{prop:global} applies to the enlarged support itself; this is support redesign, not post-hoc acceptance of an uncertified tuple. One bin corresponds to about $11.8$~m of range-equivalent uncertainty per TRP in the controlled model. The axis-aligned dilation can include tuples that no physical UE position generates, so its rate cost is conservative for this uncertainty radius.

\begin{table}[t]
\centering
\caption{Controlled support enlargement at $\Delta=100$~m. Coverage and $48$~m anomaly flags use the same $30{,}000$ held-out positions as Fig.~\ref{fig:anomaly}.}
\label{tab:support-envelope}
\scriptsize
\setlength{\tabcolsep}{2.1pt}
\begin{tabular}{@{}crcrrrr@{}}
\toprule
$q$ & $|\Ssupp^{(q)}|$ & $(\mathcal{K}_1,\mathcal{K}_2,\mathcal{K}_3)$ & $R_{\rm fun}$ & Save & Native cov. & $P_{\rm flag}(48\,\mathrm m)$\\
\midrule
0 & 12,772 & $(103,7,103)$ & 5.393 & 21.9\% & 97.90\% & 46.14\%\\
1 & 47,742 & $(101,17,96)$ & 5.777 & 16.4\% & 100\% & 25.32\%\\
\bottomrule
\end{tabular}
\end{table}

All six terminal orders are re-evaluated with the common $P=3$ budget, and both rows have zero decoder conflicts. The one-bin envelope covers every native tuple in the controlled held-out sample while retaining $16.4\%$ ideal-rate savings; its zero-bias color-tuple flag rate is also $0$ in those $30{,}000$ samples. The price is lower anomaly sensitivity: at a $48$~m bias on TRP~1, the flag probability falls from $46.14\%$ to $25.32\%$. Hence support width is a concrete compression-coverage-anomaly trade-off rather than a hidden modeling choice.

Operationally, the nominal and enlarged constructions can be provisioned as separate, globally verified codebook versions. A sustained rise in native fallback can trigger a coordinated switch through the same versioning mechanism already required by the protocol; no online unverified recoloring is needed. More adaptive policies, including learning-based selection among pre-certified profiles, are possible, but the safety invariant should remain that every activated profile has passed the complete decoder test offline.

\section{Discussion and Conclusion}
\label{sec:conclu}

Zero-error refers to the service function $\ell$, not to the true UE position. On the certified support, compressed and native reports yield the same service-grid decision. Positioning accuracy still depends on synchronization, geometry, propagation, receiver processing, and the estimator. The construction also applies to other finite local measurements, such as UL-RTOA, AoA, RSRP, or jointly quantized features. A codebook remains valid only while its TRP set, quantizer, support model, and service function remain unchanged.

We introduced a global product-partition criterion and a conditional conflict graph for support-aware distributed telemetry coding. In the controlled 3-TRP study, the method reduces the ideal raw rate by $21.3$--$22.1\%$ for service grids up to $100$~m, with zero decoder conflicts. In the measured six-base-station TDoA study, native-support recurrence is $83.77\%$. For $2$--$8$~m grids, feasible-tuple screening accepts $85.19$--$89.39\%$ of test bursts. Among the accepted bursts, the compressed/native decision mismatch rate is $0.81$--$1.05\%$, while two-stage fallback provides $8.0$--$25.0\%$ ideal measurement-rate savings. The same feasible set also provides a static consistency screen. Separately, the OAI prototype reduces the complete 3-TRP MeasurementResponse from $69$ to $38$ aligned-PER bytes and validates receive-side decoding. This byte-level result uses a temporary synthetic codebook for protocol integration and is reported separately from the scientific compression results. In the controlled holdout, a one-bin certified support enlargement raises native-tuple coverage from $97.90\%$ to $100\%$ while retaining $16.4\%$ ideal-rate savings, and quantifies the corresponding loss in anomaly sensitivity.

The measured dataset is downlink TDoA, the OAI prototype uses logical rather than synchronized physical TRPs, and DSATUR remains heuristic. Future work will couple the validated scientific codebooks to synchronized UL-RTOA, calibrate the support and fallback policy with additional traces, and measure latency and service agreement in RFsimulator and OTA experiments. An important next step is to learn uncertainty envelopes from calibration data, rather than selecting their width manually, while keeping each deployed profile independently zero-error certified.

\balance
\bibliographystyle{IEEEtran}
\bibliography{refs}

\end{document}